\documentclass[12pt]{article}
\usepackage{graphicx,psfrag}
\usepackage{color}
\usepackage{float}

\usepackage{amsmath,amsthm,amssymb}

\newtheorem{thm}{Theorem}
\newtheorem{prop}{Proposition}
\newtheorem{defn}{Definition}
\newtheorem{cor}{Corollary}
\newtheorem{lem}{Lemma}
\newtheorem{rem}{Remark}
\newtheorem{ex}{Example}

\def\ran{\rangle}
\def\lan{\langle}

\newcommand{\m}{\marginpar}
\newcommand{\sg}{\sigma}

\newcommand{\ta}{\theta}
\newcommand{\vt}{\vartheta}

\newcommand{\al}{\alpha}
\newcommand{\be}{\beta}
\newcommand{\de}{\delta}
\newcommand{\vp}{\varphi}

\begin{document}

\begin{center}
{\Large\bf Marginally outer trapped tubes in de Sitter spacetime} \\ ~\\
{\large Marc Mars} \\
Departamento de F\'isica Fundamental, Universidad de Salamanca \\
Plaza de la Merced s/n, 37008 Salamanca, Spain
\\~\\ {\large Carl Rossdeutscher, Walter Simon and Roland Steinbauer}\\
Fakult\"at f\"ur Mathematik, Universit\"at Wien\\
Oskar-Morgenstern Platz 1, 1090 Wien, Austria
\end{center}

\begin{abstract}

We prove two results which are relevant for constructing marginally outer trapped tubes (MOTTs)
in de Sitter spacetime. The first one (Theorem 1) holds more generally, namely for spacetimes
satisfying the null convergence condition and containing a timelike conformal Killing vector with a "temporal function". We show that all marginally outer trapped surfaces (MOTSs) in such a spacetime are unstable. This prevents application of standard results on the propagation of stable MOTSs to MOTTs.

On the other hand, it was shown recently \cite{HHT} that, for every sufficiently high genus, there exists a smooth, complete family of CMC surfaces embedded in the round 3-sphere $\mathbb{S}^3$. This family connects a Lawson minimal surface with a doubly covered geodesic 2-sphere. We show (Theorem 2) by a simple scaling argument that this result translates to an existence proof for complete MOTTs with CMC sections in de Sitter spacetime. Moreover, the area of these sections increases strictly monotonically. We compare this result with an area law obtained before for holographic screens.
\end{abstract}

\section{Introduction}

Our general setting is a smooth $n$-dimensional ($n \ge 3$), oriented and time-oriented Hausdorff manifold $({\cal M}, g)$ with a smooth metric of signature $(-,+,...+)$. In Sect.\ 2 we require the Ricci tensor Ric${}_g$ to satisfy the null convergence condition Ric${}_g(\ell,\ell) \ge 0$ for all null vectors $\ell$,
while in  Sect.\  3  we restrict ourselves to  de Sitter spacetime (dS).
We first recall well-known key definitions to also fix our notation. Note that we use the "physics" convention
throughout, in particular in Def.\   1 where most mathematics references (e.g.\  \cite{HHT}) define the mean curvature by $\widehat H = H/\mbox{dim}~ {\cal U}$.

\begin{defn}
\label{mc}
Let ${\cal U}$ be a hypersurface  (i.e.\ a submanifold of codimension one) embedded in any semi-Riemannian manifold ${\cal N}$. The {\bf mean curvature} $H$ of ${\cal U}$ is the trace of its second fundamental form.
\end{defn}

\begin{defn}
\label{mots}
Let  ${\cal F}$ be a spacelike surface of codimenion two embedded in ${\cal M}$ with future directed, orthogonal null
vectors $\ell^{\pm}$, and denote the null expansions of the corresponding families of emanating geodesics by $\ta^{\pm} = tr_{\cal F} (\nabla \ell^{\pm})$.

\begin{itemize}
\item A {\bf Marginally outer trapped surface (MOTS)} or just {\bf Marginal surface} ${\cal F}$ is defined by $\ta^{+} =0$. Here $\ell^+$ defines the "outer"  direction; the latter is ill-defined for minimal surfaces for which $\ta^{-} = 0$ as well. (As to terminology cf.\ Remark \ref{outer} in Sect.\ \ref{nost}).
\item A {\bf Marginally trapped surface (MTS)}
is a MOTS with  either $\ta^{-} < 0$  or $\ta^{-} > 0$ (we avoid the term "past trapped" or "antitrapped" for the latter case).
\item A {\bf Marginally outer trapped tube (MOTT) ${\cal T}$} is a hypersurface  foliated by MOTS.
\item A {\bf Marginally trapped tube (MTT)} or {\bf holographic screen} is a hypersurface foliated by MTSs.
\end{itemize}

\end{defn}

MOTS, MTS or trapped surfaces which satisfy $\ta^{+} < 0$ and $\ta^{-} < 0$, are key in  singularity theorems which yield incomplete geodesics.
However, in a typical collapse scenario, there is not only the (outermost) ``key MOTS'' but also
 a plethora of other, seemingly ``spurious''  MOTS (see e.g. \cite{PK,IB} and references therein).
One motivation of this work is to better understand such spurious  MOTS.

Apart from these collapse scenarios, MOTS are normally also abundant in non-singular spacetimes which do not satisfy all requirements for the singularity theorems.
We mention two settings in which they are worth studying:
Firstly, under additional assumptions, null geodesics emanating from MOTS have conjugate points while still being complete.
This suggests a recent  definition \cite{NY} of the "outward affine size" of a non-singular universe
with a MOTS. Secondly, for asymptotically dS spacetimes, there are results on the "visibility" of future weakly trapped surfaces (defined by $\ta^{\pm}\le 0$) from infinity \cite{CGL}.

A key property of  MOTS is their (in)stability, cf.\ Def.\ \ref{stab} and Remark \ref{var} below. Here we recall important properties of strictly stable MOTSs: They  have spherical topology, inherit all symmetries of the ambient space, and propagate along achronal MOTTs in a given spacelike foliation of spacetime \cite{AMS,AMMS}. Somewhat weaker results hold in the non-strictly stable case.
Moreover, stability also plays a role in a singularity theorem (cf.\ Thm. 7.1 and Remark 7.2 in \cite{AMMS}).

The present work has two somehwat independent aspects exposed in Sect.\  2 and 3. In Sect.\ 2 we show absence of stable MOTS in a class of spacetimes containing a timelike conformal Killing field with a "temporal function". The precise result is as follows:

\begin{thm}
    \label{thmConf}
Let  $({\cal M},g)$  be an $n$-dimensional  ($n \ge 3$) spacetime satisfying the null convergence condition. Assume  $({\cal M},g)$ admits a
    future directed timelike conformal Killing field $\xi$, i.e.
    \begin{align*}
      \pounds_{\xi} g = 2 \Psi g
    \end{align*}
    where $\pounds_{\xi}$ is the Lie derivative and $\Psi$ is a temporal function, i.e. its gradient is timelike and past directed. Then $({\cal M},g)$ admits no stable MOTS.
\end{thm}

In Sect.\ 3 we restrict ourselves to  dS for which Ric${}_g = (n - 1)  \delta^2 g$, $\delta$ a positive constant. It can be globally written in the form
\begin{align}
\label{sph}
  g_{\mbox{\tiny dS}} = \frac{1}{\delta^2 \cos^2 \sigma} \left ( - d \sigma^2 + g_{\mathbb{S}^{n-1}} \right )
\end{align}
where $\sigma \in \left ( - \frac{\pi}{2}, \frac{\pi}{2} \right )$ is taken to increase to the future, and
$g_{\mathbb{S}^{n-1}}$ is the standard metric  on the round unit sphere
$\mathbb{S}^{n-1}$.

Choosing $\xi = \frac{\partial}{\partial \sigma}$ in Thm.\  1 we show in Corollary 1 that all MOTSs in dS are unstable. Hence the above-said results of \cite{AMS,AMMS}  on topology,
symmetry and evolution to achronal MOTTs do not apply. To analyze  MOTTs we now restrict ourselves to 3+1 dimensions and focus on umbilic and constant curvature slicings of dS (with curvatures of all possible signs). It is easy to show (Lemma 2 in Sect.\ 3.2.1) that, in this setting, MOTSs are automatically MTSs and (as hypersurfaces within the slices) surfaces of constant mean curvature (CMC).
This allows, in particular, constructing a MTSs from certain families of CMC surfaces in either flat $\mathbb{R}^3$, the hyperbolic space $\mathbb{H}^3$ or the round $\mathbb{S}^3$.
We consider the former two settings (in Sects.\ \ref{flat} and \ref{hyp}) as a ``warmup" and focus on $\mathbb{S}^3$ (cf.\  Sect.\ \ref{css}).

In Sect.\ 3.2.4 we analyze, then,  dS in terms of its "standard" umbilic slicing (\ref{sph}) by round three spheres ${\mathbb S}^3_{\sg}$
of radius $\rho = \delta^{-1} \cos^{-1} \sigma$.
In this slicing, MOTS with $\ta^{\pm} = 0$ correspond to
CMC surfaces with mean curvature $H = \mp 2 \delta \sin \sigma$, thus in particular minimal for $\sigma = 0$. The resulting MTSs are easily found explicitly in the case of spherical and toroidal MTS-sections. In the former case, the MTSs are null surfaces with sections of  constant area, while in the latter case, the MTSs are timelike and their area increases monotonically from the Clifford torus at $\sigma = 0$ to $\sigma \rightarrow \pm \frac{\pi}{2}$ (where it diverges).

On the other hand, knowledge of CMC surfaces of genus ${\mathfrak g}$ larger than one in $\mathbb{S}^3$  is rather rudimentary. However, for every ${\mathfrak g}$ sufficiently large (${\mathfrak g} \gg 1$), S. Charlton, L. Heller, S. Heller and M. Traizet \cite{HHT} recently proved existence of a complete family of  CMC surfaces
$f_{\psi}^{\mathfrak g}$,  ($\psi \in (-\frac{\pi}{4}, \frac{\pi}{4}$)) embedded in the unit sphere $\mathbb{S}^3$. This family connects the well-known Lawson minimal surface $\xi_{1,{\mathfrak g}}$ (for $\psi = 0$)
with a doubly covered geodesic sphere (for $\psi = \pm \frac{\pi}{4}$).
Moreover, these authors show that the Willmore energy of these surfaces,
defined by
\begin{align}
\label{Wil}
W = A (\frac{H^2}{4} + 1),
\end{align}
where $A$ is the area, increases strictly monotonically from the Lawson surfaces towards the limiting spheres.

After quoting and slightly extending the results of \cite{HHT} as Theorem 3 in Sect.\ 3  we describe its straightforward application to dS. In fact we only need to rescale the geometry from the  $\mathbb{S}^3$ of unit radius to the spherical sections of dS whose radius blows up with the cosmological time $\sigma$ as sketched above. Remarkably, this scaling entails that the Willmore energy of the CMC surfaces found in  \cite{HHT} translates just to the area of the MTS-sections in the dS setting and yields the corresponding monotonicity statement. Thus our adaptation of the results of \cite{HHT} reads as follows.

 \begin{thm}
 \label{family}
 Our setting is de Sitter spacetime in coordinates (\ref{sph}) for $n = 4$. For every  ${\mathfrak g} \in \mathbb{N}$, ${\mathfrak g} \gg 1$ , there exist a smooth function $\sg: (-\frac{\pi}{4}, \frac{\pi}{4})\ni\psi \mapsto  \sg(\psi) \in (-\sg_{m}, \sg_{m})$ and a smooth family of conformal CMC embeddings ${\cal F}_{\psi}^{\mathfrak g}: {\cal R}^{\mathfrak g} \rightarrow \mathbb{S}^3_{\sg}$ of a Riemann surface ${\cal
 R}^{\mathfrak g}$ of genus ${\mathfrak g}$ into the round 3-sphere $\mathbb{S}^3_{\sg}$ of radius $\rho_{\sg} = \delta^{-1} \cos^{-1}\sg$ such that

 \begin{enumerate}
  \item  ${\cal F}_{0}^{\mathfrak g}$  is the Lawson surface $\xi_{1,{\mathfrak g}}$ of genus ${\mathfrak g}$.
\item For $\psi \rightarrow \pm \frac{\pi}{4}$, ${\cal F}_{\psi}^{\mathfrak g}$  smoothly converges to a doubly covered geodesic 2-sphere with $2{\mathfrak g} + 2$ branch points; the family  ${\cal F}_{\psi}^{\mathfrak g}$ cannot be extended in $\psi$ in the space of immersions.
\item ${\cal F}_{\psi}^{\mathfrak g} = {\cal F}_{- \psi}^{\mathfrak g}$ up to reparametrization and isometries of $\mathbb{S}^3_{\sigma}$,
and accordingly  the  constant mean curvatures satisfy ${\cal H}_{\psi}^{\mathfrak g} =  - {\cal H}_{- \psi}^{\mathfrak g}$.
 \item  ${\cal H}_{\psi}^{\mathfrak g}$ decreases strictly monotonically from zero
 at $\psi = 0$ to a minimum  ${\cal H}_{\psi_m}^{\mathfrak g} = -2 \delta  \sin \sg^{\mathfrak g}_m$ at some $\psi =  \psi_{m}$, from which it increases
 strictly monotonically to zero for  $\psi \rightarrow  \frac{\pi}{4}$. The monotonicity behaviour of  ${\cal H}_{\psi}^{\mathfrak g}$ for $\psi \in (-\frac{\pi}{4}, 0)$ is then determined by property 3.
\item The area of ${\cal F}_{\psi}^{\mathfrak g} $ increases strictly monotonically from
 $\psi = 0$ towards both $\psi = \pm \frac{\pi}{4}$ where it takes the values
 $ A_{\pm \pi/4} = 8 \pi \de^{-2} $.
 \end{enumerate}

\end{thm}

This theorem has the following obvious Corollary:

\begin{cor}
 For every
 ${\mathfrak g} \gg 1$, the smooth family of conformal embeddings constucted in Thm.\  \ref{family} determines a  smooth MTT whose MTS sections have properties 1-5.
\end{cor}


In the final  Sect.\ \ref{dis}
we note that the monotonicity of the  area of the MTS-sections of MTTs (holographic screens) is the key ingredient in their thermodynamic interpretation along the lines of the "second law". We briefly discuss the area
law established in \cite{BE}.


\section{A class of spacetimes with no stable MOTSs}

\label{nost}

We start with definitions, setting out from Def. \ref{mots}.  We normalize the null vectors to satisfy $\lan \ell^+, \ell^- \ran = -2$. Capital  letters ($A,B,..$) denote  indices on objects on ${\cal F}$. The induced metric on ${\cal F}$ is denoted by $j$ or $j_{AB}$, while its Levi-Civita derivative by $D$ or $D_A$ and $\Delta_{j} := D_A D^A$ is the Laplacian and $R_j$ the scalar curvature.  The null second fundamental form of ${\cal F}$ with respect to $\ell^+$ is denoted by $k^{+}$ ($k_{AB}^+$) and the torsion one form $s ~(s_A)$ is defined by
\begin{align*}
  s (X) = - \frac{1}{2} \lan \ell^{-}, \nabla_X \ell^+ \ran  \quad \forall~ \mbox{vector fields} ~X~
  \mbox{on} ~ {\cal F}.
\end{align*}

\begin{rem}
\label{analogy}
The motivation and inspiration for the following definition and the results on (in)stability of MOTS
 (cf.\ Sects.\ 3-5 of \cite{AMS}) come from corresponding material on minimal surfaces of codimension 1 in manifolds with a Riemannian metric.
 We recall, however, that in the minimal surface case stability can be defined via the second variation of the area functional. This is no longer the case for MOTS, which entails substantial
 differences (cf.\   Remark \ref{var} below).
 \end{rem}
The key definitions for this section read as follows.

\begin{defn}[\bf Stability]~ We denote the Einstein tensor by $\mathrm{Ein}_g = \mathrm{Ric}_g - \frac{1}{2}{\mathrm R}_g g$.

\label{stab}
\begin{itemize}
 \item The stability operator $L$ of a MOTS ${\cal F}$ is defined for smooth functions $w$ by
\begin{align}
\label{Lw}
  L(w) =   - \Delta_{j} w + 2 s^A D_A w +
  \frac{1}{2} \Big (\mathrm{R}_{j}  - \mathrm{Ein}_g(\ell^+,\ell^-)
  -2 s^A s_A + 2 D_A s^A \Big )  w
\end{align}
\item $L$ admits a principal eigenvalue $\lambda$, which is the eigenvalue with smallest real part. This eigenvalue is necessarily real and its eigenspace is one-dimensional and spanned by an everywhere positive function
$\phi$, called {\it principal eigenfunction}.
\item
The MOTS ${\cal F}$ is called strictly stable (resp. stable, marginally stable or unstable) provided $\lambda > 0$
($\lambda \ge 0$,  $\lambda=0$ or $\lambda <0$).

\end{itemize}
\end{defn}

\begin{rem}
\label{var}
The stability operator is relevant for calculating the variation $\delta_{\xi} \theta^{+}$  of the expansion $\theta^{+}$ with respect to an arbitrary vector field $\xi$ normal to ${\cal F}$.
In Lemma 3.1 of \cite{AMS} it was shown that
\begin{align}
  \delta_{\xi} \theta^{+} = L(u) - \frac{1}{2} B Z
  \label{delth}
\end{align}
where  the functions $u,B$ and $Z$  are defined on ${\cal F}$ by $u := \lan \xi, \ell^+ \ran$, $B := - \lan \xi, \ell^{-} \ran$ and $Z:= k^{+}_{AB} k_{+}^{AB} + \mathrm{Ein}_g (\ell^+,\ell^+)$.

The  ``stability operator $L_v$ in the direction of a normal vector $v$"  introduced in
Def.3.1 of \cite{AMS} is related to (\ref{delth}) as follows:
\begin{align}
\label{Lv}
 \delta_{\xi} \theta^+ = L_v (u) ~~\mbox{for} ~ \xi= u\, v
 \end{align}
In particular, $L(u)$ as defined in (\ref{Lw}) agrees with $L_v(u)$ for $v = - \frac{1}{2}\ell^{-}$.

\end{rem}

The stability operator of a stable MOTS satisfies a maximum principle, cf.\ Lemma 4.2 of \cite{AMS}:

\begin{lem}
  \label{MaxPrin}
  Let ${\cal F}$ be a MOTS and $L$  the corresponding stability operator. Let $\lambda$ be the principal eigenvalue and
  $\phi>0$ the principal eigenfunction. Let $\psi$ be a smooth solution of
  $L \psi = f$ with $f \geq 0$. then
  \begin{itemize}
  \item[(i)] If $\lambda =0$, then $f =0$ and $\psi = C \phi$ for some constant $C$.
  \item[(ii)] If $\lambda > 0$  and
    $f$ is not identically zero then $\psi >0$.
  \item[(iii)] If $\lambda > 0$ and $L\psi = 0$, then $\psi =0$.
  \end{itemize}
\end{lem}
From this lemma we can prove the following general result on the non-existence of stable MOTS:

\begin{prop}
  \label{deltath}
  Let ${\cal F}$ be a MOTS in an n-dimensional ($n \ge 3$) spacetime  $({\cal M},g)$ satisfying the null convergence condition. Assume that there exists a future causal vector field $\xi$ along ${\cal F}$ which is not everywhere proportional to $ \ell^+$ and  such that $\delta_{\xi} \theta^{+} \geq 0$. Then ${\cal F}$ cannot be strictly stable.
  Moreover, if $\delta_{\xi} \theta^{+}$ is positive somewhere then ${\cal F}$
  cannot be marginally stable either.
\end{prop}

\begin{proof} The basis $\{ \ell^+,\ell^-\}$ is future directed, so $\xi$ being future causal implies $u \leq 0$ and $B \geq  0$. Moreover $u$ is not identically zero, because $\xi$ is not everywhere proportional to $\ell^+$. Together with Equ. (\ref{delth}) and the  hypotheses of the proposition this implies $L(u) \geq 0$. If ${\cal F}$ is strictly stable (i.e. $\lambda >0)$, items (ii) and (iii) of Lemma \ref{MaxPrin} imply that $u$ is either strictly positive or identically zero, which is a contradition. This proves the first claim of the proposition. Finally, if ${\cal F}$ is marginally stable then  item (i) of the Lemma   yields $\delta_{\xi} \theta^{+} + \frac{1}{2} B Z =0$, which cannot happen if  $\delta_{\xi} \theta^{+}$ is positive somewhere.
  \end{proof}

 We are now ready to prove Thm.\  \ref{thmConf} stated in the Introduction. We use greek indices
 on spacetime objects and sum over repeatedly occurring ones.

\begin{proof}[Proof of Thm.\  1]
    In Corollary 1 of \cite{CM} the following general identity is
  proved for variations of $\theta^{+}$ of a MOTS along an arbitrary vector
  field $\xi$ in terms of the deformation tensor $a(\xi)
:= \pounds_{\xi} g$ of
  $\xi$
  \begin{align}
    \delta_{\xi} \theta^{+} = -\frac{1}{4} \theta^{-} a(\xi)_{\alpha\beta}
    \ell^{\alpha}_+ \ell^{\beta}_+
    - a(\xi)_{\alpha\beta} e_A^{\alpha} e_B ^{\beta}
    k_{+}^{AB}
    + j^{AB} e_A^{\alpha} e_B^{\beta} \ell^{\nu}_+
    \left (  \frac{1}{2} \nabla_{\nu} a(\xi)_{\alpha\beta}
    - \nabla_{\alpha} a(\xi)_{\beta\nu}    \right )
    \label{identity}
  \end{align}
  where $\{e^{\alpha}_A\}$ is any basis of tangent vectors to ${\cal F}$.
    Assume now that $\xi$ is a conformal Killing vector, so that
  $a(\xi) = 2 \Psi g$. Inserting into \eqref{identity} yields
  \begin{align*}
    \delta_{\xi} \theta^{+} =
    j^{AB} e_A^{\alpha} e_B^{\beta} \ell^{\nu}_+
    \left ( \nabla_{\nu} \Psi g_{\alpha\beta}
    - 2 \nabla_{\alpha} \Psi g_{\beta\nu} \right )
    = (n-2) \ell^{\nu}_+ \nabla_{\nu} \Psi.
  \end{align*}
  Since $\ell^+$ is future directed and $\Psi$ is a temporal function, we have
  $\delta_{\xi} \theta^{+} >0$ and the result is a consequence of Proposition
  \ref{deltath}.
\end{proof}

\begin{rem}
\label{outer}
Recall that MOTS were defined in  Def. \ref{mots} by the vanishing of (at least) one of the null expansions. As the  ``outer"  amendment apparently plays no role, neither in that
Definition nor in the discussion above, the simpler notion ``marginal surface'' (introduced by Hayward \cite{SH}) would  suffice. To indicate  some benefit of the ``MOTS" terminology, we recall
that in Def.3.1 of \cite{AMS}, stability of the MOTS with respect to any direction $v$
was defined  by requiring that the operator $L_v(w)$ (Equ. (\ref{Lv}))  has non-negative lowest eigenvalue. This is equivalent to non-negativity of either side of (\ref{Lv}) provided
$\ell^+$ points in the same direction as  $\xi = u\,v$, viz. $u = \lan \xi, \ell^{+} \ran \ge 0$. Calling these directions  ``outer"  (or ``inner")  is just more  efficient terminology
than some ``... points in the same (or opposite) direction..." amendment. The next section provides corresponding ``unstable examples".

\end{rem}

\begin{ex}[The Friedman-Lemaitre-Robinson-Walker Universe (FLRW)]

Before focusing on dS in the next section, we apply Theorem 1 to n-dim. FLRW, viz.

\begin{equation}
 ds^2 = -dt^2 + a(t)^2 \left( \frac{dr^2}{1 - \kappa r^2} + r^2 d\Omega^2 \right)
 \end{equation}
\noindent
 where $d\Omega^2$ is the spherical metric in $n-2$ dimensions and $\kappa \in \{-1,0,1\}$.

The requirements that $\Psi$ is  a temporal function for the conformal
Killing field $\xi=a \partial_t$ and the null convergence condition translate into the
following conditions on the scale factor $a(t)$, respectively

\begin{equation}
\label{adot}
\ddot{a} >0,  \qquad \ddot{a} \leq \frac{\dot{a}^2 + \kappa}{a}.
\end{equation}

Here  dot is the derivative with respect to $t$. If conditions (\ref{adot}) hold, all MOTS are unstable.

We remark in this context that MOTS in FLRW were investigated e.g. in \cite{FHO,MX}. In the latter paper it was found that in closed ($\kappa = 1$) 4-dim. FLRW, a certain family of MOTS called CMC Clifford tori are always unstable. In the next section we will analyze systematically
families of CMC-MOTS in umbilic slices of dS.

\end{ex}

\section{MOTSs and MOTTs in de Sitter spacetime}

\subsection{Instability of MOTS}

The following easy consequence of Thm.\  \ref{thmConf} was already sketched in the Introduction.
\begin{cor}
\label{dsunst}
  All MOTS in de Sitter spacetime are unstable.
\end{cor}
\begin{proof}

 In terms of the coordinates (\ref{sph}), $\xi = \frac{\partial}{\partial \sg}$ is a future directed, timelike conformal Killing vector field  satisfying
  \begin{align*}
    \pounds_{\xi} g_{\mbox{\tiny dS}} = 2 \tan \sigma \, g_{\mbox{\tiny dS}}.
  \end{align*}
   Since $\tan$ is an increasing function and
  $\sigma$ is a temporal function, so is $\Psi = \tan \sigma$. Hence Thm.\
  \ref{thmConf} implies the claim.
\end{proof}

\begin{rem}
\label{anal}
There is a certain analogy between de Sitter spacetime on which we focus in this section and round spheres, in the sense that both are maximally symmetric spaces of constant positive curvature.
Recalling also the analogy between the stability definitions for minimal surfaces and MOTS (cf.\  Remark \ref{analogy}),  Corollary \ref{dsunst} above can be seen as a counterpart to a result of  \cite{JS} (namely case $p = n-1$ of Thm 5.1.1)  that all minimal surfaces of codimension 1 on the round ${\mathbb S}^n$ are unstable. However,
there is no obvious analogy between the proofs, and we do not make any attempt of establishing
one here.
\end{rem}

\subsection{MOTS in umbilic slicings}

From now on we restrict ourselves to dS in 3+1 dimensions.
Our aim is to locate MOTTs in the three umbilic and  constant curvature slicings of de Sitter, namely flat and hyperboloidal and spherical. Our  focus will be on the latter case.

\subsubsection{Preliminaries}

We consider a spacelike hypersurface $({\cal N}, h, K)$ with metric $h$, scalar curvature ${}^3 R$, covariant derivative  ${}^3 \nabla$,  extrinsic curvature $K$ (w.r.t.\ a future-pointing unit normal $N$) and mean curvature $tr_h K$  (i.e.\ again with "physics convention", cf.\ Def.\ \ref{mc}). The constraints take the form
\begin{eqnarray}
 R_h - |K|_h^2 + (tr_h K)^2 &= &6 \delta^2  \label{ham}  \\
div_h \left(K  - (tr_h K) h  \right) & = & 0.   \label{mom}
\end{eqnarray}

\begin{lem} \label{umb} Let $({\cal N}, h,K)$ be an umbilic slice in de Sitter, i.e. the extrinsic curvature satisfies $K = \beta h$
for some function $\beta$. Then $\beta$  is constant on ${\cal N}$. Furthermore,
 any MOTS ${\cal F} \subset {\cal N}$ with
 $\theta^{+} = 0$  is a MTS and CMC as well, with mean curvature $H = div_{h} X = - 2 \beta$
where $X$ is a unit outward normal vector to ${\cal F}$ in ${\cal N}$, i.e.\ $\lan X, \ell^+ \ran > 0$.
\end{lem}

\begin{proof}

The constancy of $\beta$ follows from the constraint (\ref{mom}). For the second statement, we use the normalisation $\lan \ell^+, \ell^- \ran = -2$ and the decomposition $\ell^{\pm} = N \pm X$. This induces the following decomposition of  $\theta^{\pm}$
\begin{equation}
\label{ta3}
\ta^{\pm}  =  \pm H + tr_{j} K =\pm  H     +  2 \beta
\end{equation}
which implies the assertion, with  $\ta^{-} = 4 \beta$ on the MOTS $\theta^{+} = 0$.
\end{proof}

\begin{rem}
\label{bcm}
In the following subsections we consider as examples umbilic slicings $({\cal N}, h,K)$ of constant curvature (i.e. $h$ is flat, spherical or hyperboloidal) in dS. In this special situation, a recent result of \cite{BCM} implies
instability of all MOTS ${\cal F}$ contained in such slices. The argument goes as follows.
The spaces ${\cal N}$ in question are homogeneous in the sense that their isometry groups
act transitively. In particular, for a given MOTS ${\cal F}$ any pair of points
$p \in{\cal F}$ and $q \not\in {\cal F}$ can be joined by a group element $\rho$, i.e.  $q = \rho(p)$.
Hence the  Killing vector $\xi$ of ${\cal N}$ tangent to this orbit cannot be everywhere
tangent to ${\cal F}$. Now by Thm.\  1.6 of \cite{BCM} or by Thm. 8.1 of \cite{AMS}, ${\cal F}$ is either unstable or marginally stable, with $\xi$  nowhere tangent to ${\cal F}$ in the latter case. But
this case is ruled out by Thm. 5.1 of \cite{BCM} since by Lemma \ref{umb} above, $H = -2 \beta =$ const. on an  umbilic slice.
\end{rem}

While the results of \cite{BCM} also  cover other situations than the aforementioned one
(as discussed in that paper), there is not much further overlap with the  general settings considered in Proposition \ref{deltath} and Corollary \ref{dsunst}. Of course, the latter results apply to umbilic slices.

These slicings are conveniently defined in terms of the dS hyperboloid
\begin{align}
\label{hyper}
 -x_0^2 +  x_{\al} x_{\al} = \delta^{-2}  \qquad \al,\be \in \{ 1,2,3,4\}.
\end{align}
embedded in  5d Minkowski space
\begin{align}
ds^2 = -dx_0^2 + dx_{\al} dx_{\al}
\end{align}
 Above and henceforth, repeated indices are summed over. A useful reference
 for the slicings considered below is Sect. 4  of \cite{GP}.
\subsubsection{Flat slicing}
\label{flat}
In terms of the coordinate transformation
\begin{eqnarray}
 x_0 &=& \delta^{-1} \sinh \delta t + \frac{\delta}{2} r^2 e^{\de t} \\
 x_1 &=& \de^{-1} \cosh \de t - \frac{\delta}{2} r^2 e^{\de t}\\
 x_i &=& e^{\de t} y_i \qquad (i \in \{2,3,4\}) \qquad r^2 = y_iy_i
\end{eqnarray}
and upon introducing polar coordinates on $\mathbb{S}^2$, the induced metric on the hyperboloid (\ref{hyper}) reads
\begin{equation}
\label{fs}
ds^2 = -dt^ 2 + e^{2 \de t} [dr^2 + r^2(d\vt^2 + \sin^2 \vt d\phi^2)] \quad r \in [0,\infty),~ \vt \in [0,\pi], ~\phi \in [0,2\pi]
\end{equation}    
Clearly, the  $t = $const. [$t \in (-\infty, \infty)$] slicing is flat and umbilic (with $\be = \de$) but covers only half of dS as described by (\ref{sph}).

In order to determine the MOTS  in these slices, we recall that closed
CMC surfaces embedded in flat 3-space must necessarily be round spheres \cite{ADA}.

A sphere of  radius $r=R$ in a  slice $t= T$ obviously has the induced metric
\begin{equation}
ds_R ^2 = e^{2 \de T} R^2(d\vt^2 + \sin^2 \vt d\phi^2).
\end{equation}
We now choose the counter-intuitive label   {\it ``outgoing''}  for
the normal vector  $X$ pointing towards {\it decreasing} $r$, viz.:
\begin{equation}
\label{no}
X = - e^{ - \de T} \frac{\partial}{\partial r}.
\end{equation}
We find for the mean curvature w.r.t.\ $X$
\begin{equation}
H = div_h X=  - \frac{2}{R} e^{-\de T}.
\end{equation}
  It is only with this choice that (\ref{ta3}) can be solved to determine the location of a marginally {\it outer} trapped tube with MOTS sections $\theta^+ = 0$, namely:

   \begin{equation}
     e^{-\de T}  =   \de  R    \label{mott}
    \end{equation}
    in consistency with Def.\   \ref{mots}. We note that $H = - 2\de = const.$ along the MOTT.



\subsubsection{Hyperboloidal slicing}
\label{hyp}

We perform the coordinate transformation
\begin{eqnarray}
 x_0 &=& \delta^{-1} \sinh \delta t  \cosh r\\
 x_1 &=& \de^{-1} \cosh \de t \\
 x_i &=& \de^{-1} z_i \sinh \de t \sinh r \qquad (i \in \{2,3,4\}) \qquad  z_iz_i = 1
\end{eqnarray}
which yields for the induced metric on the hyperboloid (\ref{hyper})
\begin{align}
\label{hs}
ds^2 = -dt^2 + \delta^{-2} \sinh^2 (\delta t) [dr^2 + \sinh^2 r (d\vt^2 + \sin^2
\vt d\phi^2)]
\end{align}
where $t\in(-\infty,\infty)$, $r \in [0, \infty)$ and $\vt \in [0,\pi], ~\phi \in [0,2\pi)$.
This slicing is umbilic with $\be = \delta \coth \delta T$ for $t = T =$const.;
again it covers only part of dS given by (\ref{sph}), cf. Sect. 4.4.2 of \cite{GP}.

Here we restrict ourselves to calculating round CMC spheres; as to more general shapes cf.\  Remark \ref{Lc} below. Then the discussion becomes quite similar to the flat case of the previous subsection. The induced metric on slices $t = T =$const., $ r =R =$const. reads
\begin{align*}
ds_R^2 =  \delta^{-2} \sinh^2 (\delta T)  \sinh^2 R (d\vt^2 + \sin^2
\vt d\phi^2)
\end{align*}

We now take as {\it outgoing} unit normal vector to the $R= const$ slices
\begin{align*}
 X = - \delta \sinh^{-1} (\delta T) \frac{\partial}{\partial r}
 \end{align*}
 which yields for the mean curvatures of the MOTSs,
\begin{equation}
H =  div_h X=  - 2 \delta \sinh^{-1} (\delta T) \coth R.
\end{equation}

It then follows from (\ref{ta3}) that the  MOTT with spherical sections
  are determined by the equation
   \begin{equation}
   \cosh \delta T =  \coth R   \label{mott}.
    \end{equation}

For the mean curvatures of the MOTS sections of these  MOTTs we obtain
\begin{align}
\label{Hbd}
H = -2 \de \cosh R \le -2 \de
\end{align}

\begin{rem}
 \label{Lc}
 By the maximum principle, the mean curvature of a closed CMC surface embedded in
 the unit hyperboloid $\mathbb{H}^3$ must satisfy either $H \ge 0$ or $H \le -2$ (cf.\  e.g. \cite{DIK}), in consistency with (\ref{Hbd}). We also recall
 the ``Lawson correspondence" \cite{HBL,DIK}, which is a locally bijective relation between CMC surfaces in 3-dim space forms. In particular, CMC surfaces in the unit hyperboloid $\mathbb{H}^3$ with $H \le -2$ have corresponding CMC or minimal surfaces in flat space or in the unit sphere $\mathbb{S}^3$.
 We  restrict the detailed discussion to the latter case.
\end{rem}

\subsubsection{Complete spherical slicing}

\label{css}

Introducing  the coordinates
\begin{eqnarray}
 x_0 &=& \delta^{-1} \sinh \delta t\\
 x_\alpha &=& \de^{-1} z_{\al} \cosh \de t  \qquad (\al \in \{1, 2,3,4\}) \qquad  z_{\al}z_{\al} = 1
 \end{eqnarray}
we obtain
\begin{align}
ds^2 = -dt^2 + \delta^{-2} \cosh^2 (\delta t) [d\tau^2 + \sin^2 \tau (d\vt^2 + \sin^2
\vt d\phi^2)] \qquad \tau \in [0, \pi]
\end{align}
This is equivalent to (\ref{sph}) where the time coordinates are related by
$\tan \frac{\sigma}{2} = \tanh \frac{\delta t}{2} $. We proceed the discussion in terms of (\ref{sph}).

The   extrinsic and mean curvatures  of the surfaces $\sg = const$ read
\begin{equation}
 K =  (\de \sin \sg) \, h \qquad tr_h K =  3 \de \sin \sg.
\end{equation}

Therefore $\be = \de \sin \sg$ in the notation of Lemma \ref{umb}. In contrast to the previous slicings Sects.\ \ref{flat} and \ref{hyp}, we can now solve (\ref{ta3}) for either sign which yields
\begin{align}
\label{hcmc}
H = \mp 2 \de \sin \sg.
\end{align}

\noindent {\bf 3.2.4.1 Spherical MOTS} \\~

In this and the following subsection 3.2.4.2 our presentation largely follows \cite[Sect.\ 5]{BBS} and \cite[Ch.\ 5]{RB}.

In polar coordinates (\ref{sph}) takes the form
\begin{equation}
ds^2  =  \de^{-2}\cos^{-2}\sigma\left[ d\tau^ 2  + \sin^2 \tau \left(d \vt^2 +
\sin^ 2 \vt d\varphi^2 \right) \right];  ~~\tau, \vt \in [0, \pi], \varphi \in [0, 2\pi)
\end{equation}
For round 2-spheres  $\tau = const.$ we obtain from (\ref{ta3}):
\begin{equation}
\label{msph}
\theta^{\pm}  = 2 \de [\pm \cos \sigma \cot \tau + \sin \sigma].
\end{equation} 
 
The MOTSs  given by $\tau_0^{\pm} = \pm \sigma_0 + \frac{\pi}{2}$
are spheres with radius $\de^{-1}$.

 \begin{figure}[h!]

\begin{psfrags}
\psfrag{b}{\Huge $\bullet$}
\psfrag{tau}{\Huge $\tau$}
\psfrag{si}{\Huge $\sigma$}
\psfrag{Tp}{\Huge ${\cal T}^+$}
\psfrag{Tm}{\Huge ${\cal T}^ -$}
\psfrag{pi}{\Huge $\pi$}
\psfrag{pph}{\Huge $\pi/2$}
\psfrag{mph}{\Huge $-\pi/2$}
\psfrag{0}{\Huge $0$}
\psfrag{Jp}{\Huge ${\cal J}^+$}
\psfrag{Jm}{\Huge ${\cal J}^-$}
\includegraphics[angle=0,totalheight=5cm]{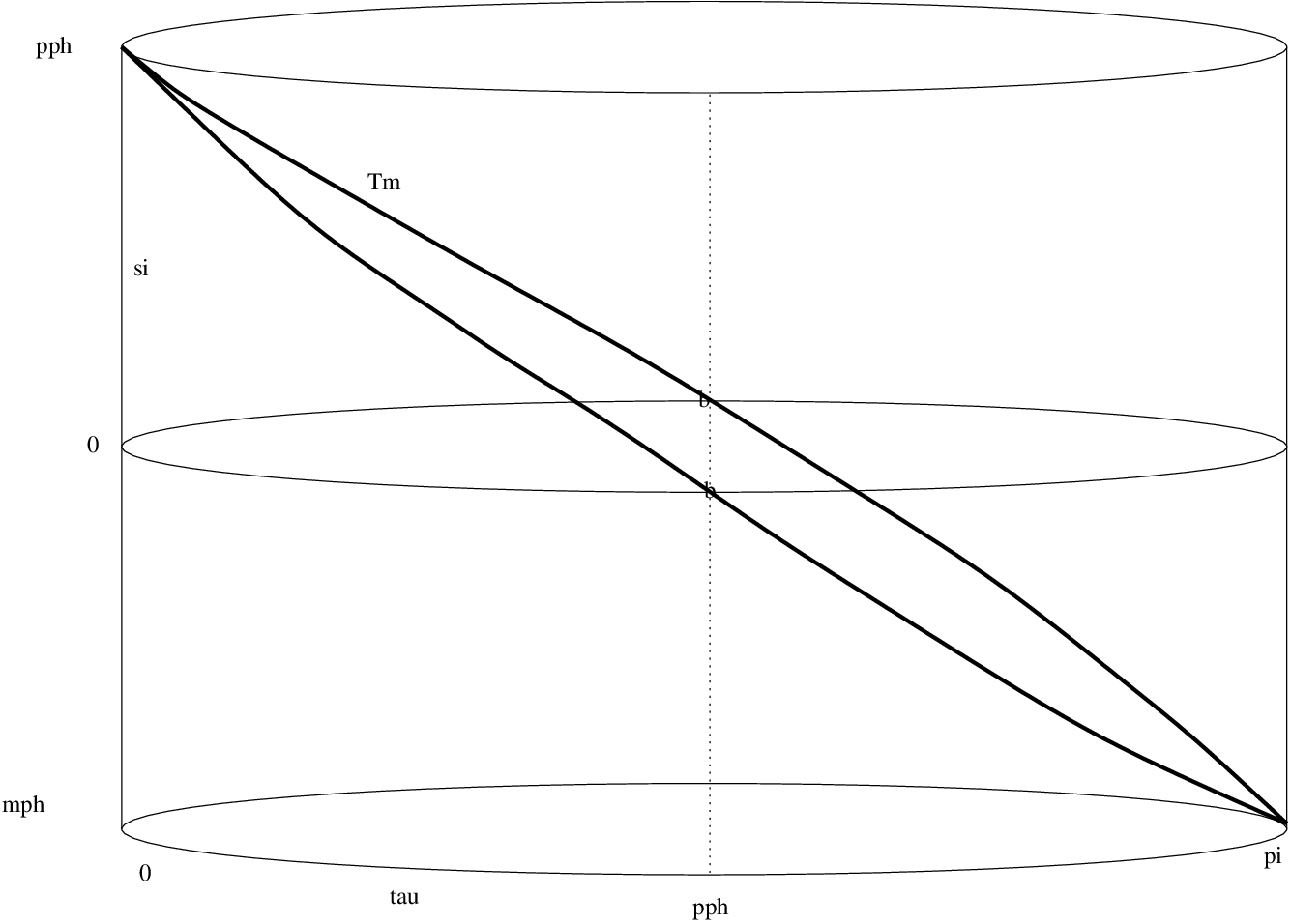}
\hspace*{1.5cm}
\includegraphics[angle=0,totalheight=5cm]{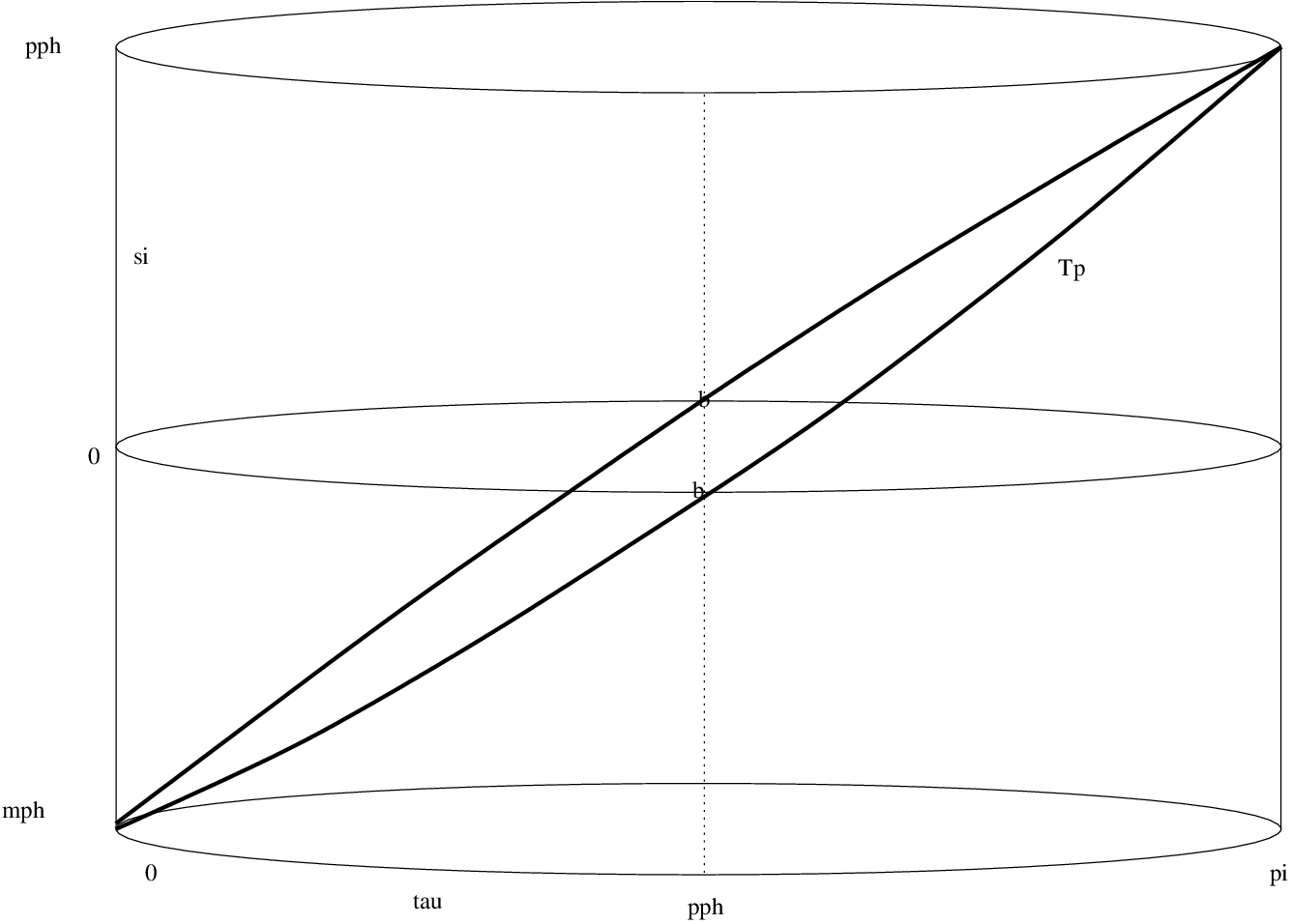}
\caption{The MOTTs ${\cal T}^-$ (left) and ${\cal T}^+$ (right) with spherical CMC sections}
 \label{sphsect}
\end{psfrags}
\end{figure}

Fig. 1 shows, via  a conformal rescaling ("Einstein cylinder", conformal factor $\de^{-2} \cos^{-2}\sigma$) the MOTTs
${\cal T}^{\pm}$ determined by the two families of MOTSs. ${\cal T}^{\pm}$  are null surfaces. The  $\varphi$ and $\vt$ directions are suppressed, hence spheres reduce to two points; the  dots at $\sg = 0$ correspond to the equators on $\mathbb{S}^3$ while the conical convergence towards $\sg = \pm \frac{\pi}{2}$ is due to the conformal factor. The drawing is sketchy (i.e.\ not a faithful result of calculation). \\

\noindent {\bf 3.2.4.2 Toroidal MOTS}\\

We recall the toroidal foliation of $\mathbb{S}^3$ in terms of  the following coordinates.
\begin{equation}
ds^2 = \de^{-2}\cos^{-2} \sigma (d \tau^2 + \sin^2 \tau \, d \gamma^2 + \cos^2 \tau \,d\xi^2 ); \quad \tau \in [0, \frac{\pi}{2}], \quad
\gamma, \xi \in [0, 2\pi)
\end{equation}

On a torus  ${\cal F}$ given by  $\tau = const.$  we find from (\ref{ta3}):
\begin{equation}
\theta^{\pm}  = 2 \de [\pm  \cos \sigma \cot 2  \tau      + \sin \sigma ]
\end{equation} 
 which formally differs  just by a factor of $2$ in the $\cot$-term from the spherical case (\ref{msph}).
This entails that the MOTSs determined by $\theta^{\pm} = 0$ and given via $\tau = \tau^{\pm}$
on some slice $\sigma = \sigma_0$
are now tori located at
$ \tau_0^{\pm}= \pm  \frac{\sigma_0}{2} +  \frac{\pi}{4}$.
As above the  MOTS exist  for all  $\sigma \in  (- \frac{\pi}{2}, \frac{\pi}{2})$ but
{\it now  they  form  timelike MOTTs.}
The area  of its toroidal sections
\begin{equation}
 \frac{A}{4\pi} = \vert \frac{\sin \tau \cos \tau }{\de^2 \cos^2 \sg} \vert = \vert \frac{\sin 2 \tau}{2 \de^2 \cos^2 \sg} \vert =\vert \frac{\sin (\pm \sg + \pi/2) }{2 \de^2 \cos^2 \sg} \vert = \frac{1}{2 \de^2 \cos \sg}
\end{equation}
diverges when $\sg \rightarrow \pm \frac{\pi}{2}$ and approaches the minimum $\frac{1}{2 \de^2}$ at $\sg = 0$.

\begin{figure}[h!]
\begin{psfrags}
\psfrag{b}{\Huge $\bullet$}
\psfrag{tau}{\Huge $\tau$}
\psfrag{si}{\Huge $\sigma$}
\psfrag{pph}{\Huge $\pi/2$}
\psfrag{mph}{\Huge $-\pi/2$}
\psfrag{pv}{\Huge $\pi/4$}
\psfrag{Tp}{\Huge ${\cal T}^+$}
\psfrag{Tm}{\Huge ${\cal T}^ -$}
\psfrag{pi}{\Huge $\pi$}
\psfrag{0}{\Huge $0$}
\psfrag{lm}{\Huge$\mathbf  \ell^-$}
\psfrag{lp}{\Huge$\mathbf \ell^+$}
\psfrag{Jp}{\Huge ${\cal J}^+$}
\psfrag{Jm}{\Huge ${\cal J}^-$}
\includegraphics[angle=0,totalheight=5cm]{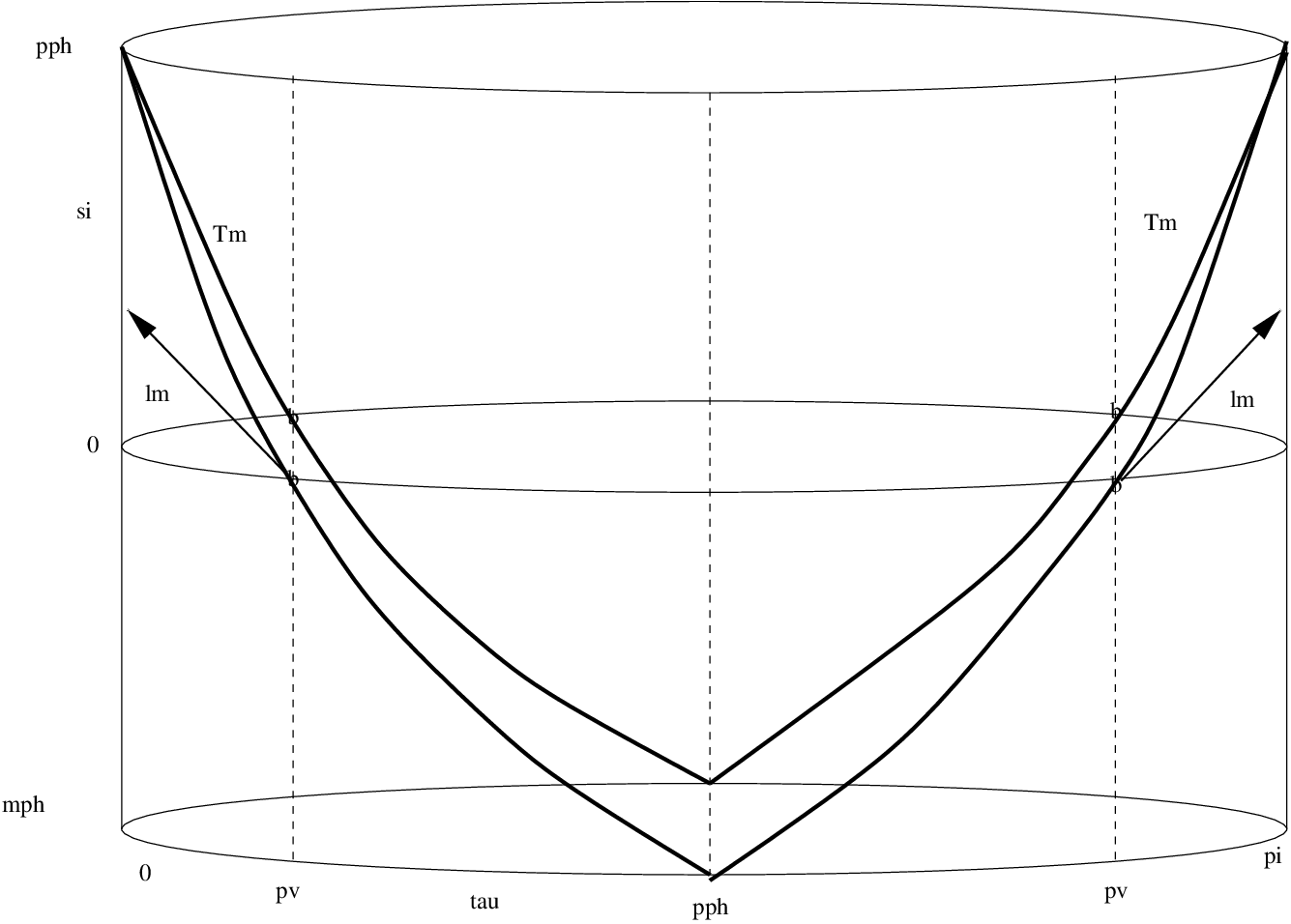}
\hspace*{1.5cm}
\includegraphics[angle=0,totalheight=5cm]{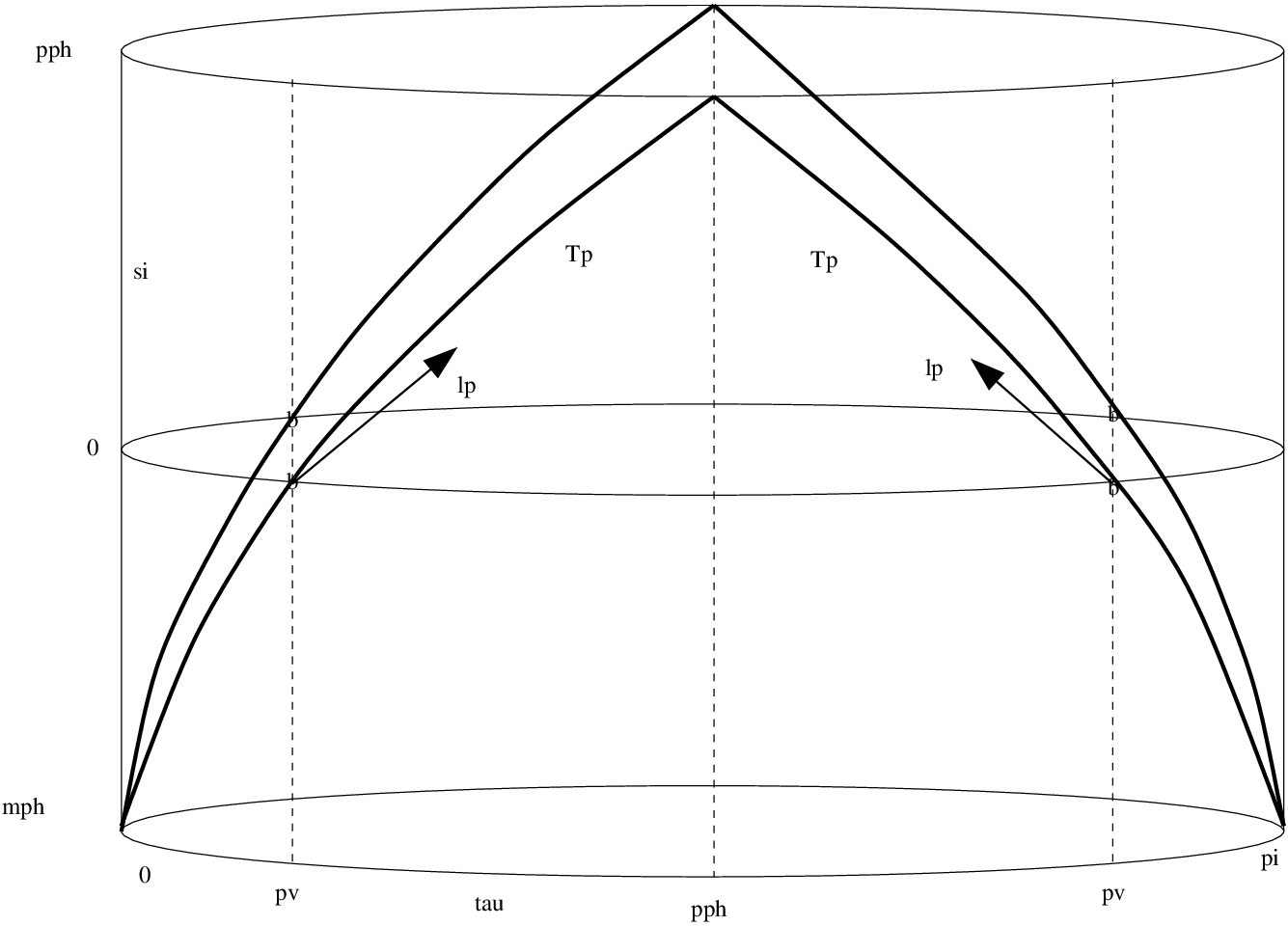}
\end{psfrags}
\caption{The MOTTs ${\cal T}^-$ (left) and ${\cal T}^+$ (right) with toroidal CMC sections}
\label{torsect}
\end{figure}

Fig. 2  shows  again the Einstein cylinder  with two dimensions suppressed;
Each MOTS ${\cal T}^{\pm}$ reduces to 4 points; the Clifford torus is indicated by 4
dots at $\sigma = 0$. The MOTSs are timelike, which is vizualized
by the null vectors $\ell^{\pm}$. As before this figure is also sketchy. \\

\noindent {\bf 3.2.4.3 MOTS of higher genus} \\

We now turn to MOTSs and MOTTs of higher genus ${\mathfrak g}$. A natural strategy is to set out from a minimal surface at the time-symmetric slice  $\sg = 0$, as such surfaces also have played special roles in the cases of  genus 0 and 1. We would like to show  existence of their evolution to MOTTs, and determine the causal character of the latter as well as the evolution of the area of their MOTS sections. Unfortunately, for genus ${\mathfrak g} > 1$ explicit expressions for the embedding functions and other geometric quantitites seem not to be available. Moreover, instability of the MOTS prevents
us from applying the general results of \cite{AMS,AMMS}.

Following \cite{HHT} we adopt here the following strategy.
It is known that any CMC surface in $\mathbb{R}^3$ or $\mathbb{S}^3$ gives a solution to a certain Lax pair; vice versa one can construct CMC surfaces as solutions to certain Lax pairs using the Sym-Bobenko formula. Making use of this, the DPW method \cite{DPW} is a way of constructing CMC surfaces from holomorphic potentials on Riemann surfaces by solving a certain Cauchy problem. To ensure that the resulting CMC surface is well-defined one has to solve in addition the so called monodromy problem. In this way the authors of \cite{HHT} construct families of CMC embeddings on $\mathbb{S}^3$ for every sufficiently high genus ${\mathfrak g}$; the latter restriction comes from the use of an implicit function theorem argument near $t = 0$, where $t = \frac{1}{2{\mathfrak g}+2} {\cal K}^{-1/2}$  and
the quantity ${\cal K}$ arises in the monodronomy construction, cf. Prop. 9 of \cite{HHT}.

We present in Thm.\  3 below   a combination of Thm.\  1 and a part of Thm.\ 2  of \cite{HHT}; item 4.\  below also  contains a slight extension of the aforementioned results.
Moreover, we change  the parametrisation from $\vp$ to $\psi =  \vp - \frac{\pi}{4}$ which simplifies the presentation.

 \begin{thm}[\bf Thms.\ 1 and  2 of \cite{HHT}]
 \label{thmunit}
 Our setting is the 3-dim. unit sphere $\mathbb{S}^3$. For every ${\mathfrak g} \in \mathbb{N}$ sufficiently large, there exists a smooth family of conformal CMC embeddings
 $f_{\psi}^{\mathfrak g} : {\cal R}^{\mathfrak g} \rightarrow \mathbb{S}^3$ from a Riemann surface with genus $g$ and parameter $\psi \in (-\frac{\pi}{4}, \frac{\pi}{4})$
 satisfying

 \begin{enumerate}

  \item  $f_{0}^{\mathfrak g}$ is the Lawson surface  $\xi_{1, {\mathfrak g}}$ of genus ${\mathfrak g}$.

  \item For $\psi \rightarrow \pm \frac{\pi}{4}$ the embedding $f_{\psi}^{\mathfrak g}$  smoothly converges to a doubly covered geodesic 2-sphere with $2{\mathfrak g} + 2$ branch points, i.e. the family $f_{\psi}^{\mathfrak g}$ cannot be extended in the parameter $\psi$ in the space of immersions.

 \item $f_{\psi}^{\mathfrak g} =  f_{- \psi}^{\mathfrak g}$ up to reparametrization and (orientation reversing) isometries of $\mathbb{S}^3$, and accordingly the constant mean curvatures satisfy
  $H_{\psi}^{\mathfrak g} = -H_{-\psi}^{\mathfrak g}$.

\item The mean curvature  $H_{\psi}^{\mathfrak g}$  of  $f_{\psi}^{\mathfrak g}$ decreases strictly monotonically from zero at $\psi = 0$ to some minimal value $ H_{\psi_m}^{\mathfrak g}$ for $\psi = \psi_m$ from which it  increases  strictly monotonically to zero for  $\psi \rightarrow  \frac{\pi}{4}$. The monotonicity behaviour for $\psi \in (-\frac{\pi}{4}, 0)$ then follows from property 3.

\item The Willmore energy Equ.\ (\ref{Wil}) of $f_{\psi}^{\mathfrak g}$ increases strictly monotonically from $\psi = 0$ towards  both $\psi \rightarrow  \pm \frac{\pi}{4}$ where
$W_{\pm \pi/4} = 8 \pi$.

 \end{enumerate}

\end{thm}

\begin{proof}

Items 1 - 3  and 5 follow as in the proofs of Thms.\ 1 and 2 of \cite{HHT}; to show item 4 we note that
 from Proposition 34 of \cite{HHT}, the mean curvature $H^{\mathfrak g}_{\psi} = H(t({\mathfrak g}), \psi)$ with
 $ t ({\mathfrak g})= \frac{1}{2{\mathfrak g} + 2} {\cal K}^{-1/2}(t,\psi)$ is smooth and its Taylor expansion in $t$ takes the form
 \begin{align}
  H(t,\psi) =  4 t \cos 2\psi \ln \left[\tan (\psi + \frac{\pi}{4})\right] + O(t^2).
  \end{align}

Here ${\cal K}$ depends on $t$ itself but such that ${\cal K}(t,\psi) = 1 + O(t^2)$ near $t = 0$, cf Sect. 6.3 of \cite{HHT}.
 A calculation shows that at the values  $\psi = \psi_0$ and $ \psi = -\psi_0$,
 $H_0 (\psi)= (\partial H/\partial t)(0,\psi)$ takes on its unique non-degenerate minimum and maximum, respectively, i.e.
 \begin{align}
 \label{ex}
 \frac{dH_0}{d\psi}(\pm \psi_0) =   0  \qquad \frac{d^2H_0}{d\psi^2} (\pm \psi_0) \neq 0,
 \end{align}
 where $\pm \psi_0$ are given implicitly by
 \begin{align}
  (\sin 2 \psi_0) \ln\left[\tan \left(\psi_0 + \frac{\pi}{4}\right)\right] = 1.
  \end{align}

 Moreover, there hold the monotonicity properties
 \begin{eqnarray}
 \label{gr}
 \frac{dH_0}{d\psi} & > &  0 \qquad  \forall~ \psi \in \left[\left.- \frac{\pi}{4}, -\psi_0\right)\right. ~
 \mbox{and}~ \psi \in \left.\left( \psi_0, \frac{\pi}{4}\right.\right] \\
 \label{sm}
  \frac{dH_0}{d\psi} &  < &  0 \qquad \forall~ \psi \in ( -\psi_0, \psi_0)
 \end{eqnarray}

We now observe  that there exists a $\psi_m(\psi, t) = \psi_0 + O(t)$ such that, for sufficiently small $t$,   properties (\ref{ex}), (\ref{gr}) and (\ref{sm}) hold for $H(t,\psi)$ as well, with $\psi_0$ replaced by $\psi_m$ everywhere. Away from the zeros of $H$ which are
at $\psi = 0,\pm \frac{\pi}{4}$ (cf.\  Prop. 18 of \cite{HHT}) this follows from smoothness and from the non-degeneracy (\ref{ex}), while in a neighbourhood of the  zeros of $H$, the assertion holds by virtue of  Prop. 34 of \cite{HHT}.
\end{proof}

\begin{rem}
 The convergence properties of the CMC surfaces for $\psi \rightarrow \pm \frac{\pi}{4} $
 (point 2 of Theorem \ref{thmunit}) are explained in Prop. 26 of \cite{HHT}.
\end{rem}

We are now ready to prove Thm.\  \ref{family} stated in the Introduction.

\begin{proof}[Proof of Thm.\  \ref{family}]
We define the conformal rescaling ${\cal I}_{\sg}: \mathbb{S}^3 \rightarrow \mathbb{S}^3_{\sg}$ from the unit sphere to the sphere of radius $\rho = \delta^{-1} \cos^{-1} \sg$. This clearly induces a scaling for all embedded surfaces; we note that the mean curvature of such surfaces
changes from $H$ to  ${\cal H}_{\sg} =  \delta H \cos \sg$.
We now define the embedding ${\cal F}_{\psi}^{\mathfrak g} : {\cal R}^{\mathfrak g} \rightarrow \mathbb{S}^3_{\sg}$ by
 ${\cal F}_{\psi}^{\mathfrak g} = {\cal I}_{\sg(\psi)} \circ f_{\psi}^{\mathfrak g}$ where $\sg(\psi)$ is constructed as follows. We recall from (\ref{hcmc}) that any CMC surface with mean curvature ${\cal H}_{\sg}$ on $\sg =$ const. corresponds to a MTS if ${\cal H}_{\sg} = \mp 2 \delta \sin \sg$.
 In particular, if the mean curvatures $H_{\psi}^{\mathfrak g}$ of the CMC surfaces
 $f_{\psi}^{\mathfrak g}$ constructed in Thm.\  3 satisfy
\begin{align}
\label{mcscal}
H_{\psi}^{\mathfrak g} = \frac{{\cal H}_{\sg}}{\delta \cos \sg} = \mp 2 \tan \sg
\end{align}
they build up  MTTs, which serves to define $\sg({\psi})$.

In order to prove the final statement 5.\ on monotonicity of the area, we show  that the  Willmore functional $W_{\sg}$ (\ref{Wil}) built from the rescaled quantities is just the area of
the embedded MOTS. To see this we note that at time $\sg$, the rescaled area is ${\cal A}_{\sg} = A \delta^{-2} \cos^{-2} \sg$  while the rescaled mean curvature is  ${\cal H}_{\sg} =  \delta H \cos \sg$. Inserting in (\ref{Wil}) we obtain $W_{\sg} = \de^2 A_{\sg}$ which finishes the proof.
\end{proof}



\begin{rem}
It is stated below Thm.\  1 in \cite{HHT} that "the moduli space of genus ${\mathfrak g}$ CMC surfaces is 1-dimensional at the Lawson surface $\xi_{1,\mathfrak g}$". This is not to be understood in the sense that the  MTTs constructed in Thm. \ref{family} above are unique,
as also indicated in point 3. of that Theorem. In fact isometries of $\mathbb S^3$ which are not tangent to  ${\cal F}_{\psi}^{\mathfrak g}$ (and which exist, cf.\ Remark \ref{bcm}) will ``move around" ${\cal F}_{\psi}^{\mathfrak g}$ at any parameter value $\psi$.
\end{rem}

\begin{rem}
We recall  that there are no non-trivial isometries of ${\cal F}_{\psi}^{\mathfrak g}$  in the case ${\mathfrak g} \ge 2$ in which the Euler number $\chi = 2(1-\mathfrak g)$ is negative. Assuming the contrary,  the Poincar\'e-Hopf theorem implies that $\chi$ is also the sum of all indices of the (isolated) zeros of the corresponding Killing field. But if an isometry    in 2 dim. has fixed points, it is necessarily a rotation in a neighbourhood. Hence all Killing indices are $+1$  which implies that  $\chi$
is non-negative. This restricts the topology of ${\cal F}_{\psi}^{\mathfrak g}$ to the sphere or the torus if it has isometries.
\end{rem}

\begin{rem}
We note that the conformal rescaling ${\cal I}_{\sg}: \mathbb{S}^3 \rightarrow \mathbb{S}^3_{\sg}$
involved in passing from Thm.\  \ref{thmunit} to Thm.\  \ref{family} is precisely the inverse of the one used to construct the Penrose diagrams Figs. \ref{sphsect} and \ref{torsect} of dS.
Hence a Penrose-type diagram arises just by directly ``stacking"
the family of CMC surfaces $f_{\psi}^{\mathfrak g}$  constructed in Thm.\  \ref{thmunit}.
We recall, however,
that the embedding parameter $\psi$ is not a monotonic function of the cosmological time $\sg$ -
rather, the constructed MOTT ``turns around" at $\psi = \pm \psi_m$.
\end{rem}

\section{Discussion}
\label{dis}
Note that monotonicity of area holds along MTTs with MTS sections of spherical, toroidal and high-genus topology in the complete spherical slicing of dS as described in  Sect.\ \ref{css}. Except for spherical MTS whose area stays constant, the area increase is even strictly monotonic
upon moving away from the time-symmetric surface.
In view of the extensively discussed correspondence between area and entropy, our MTTs are thus candidates for satisfying the ``second law" of thermodynamics in an appropriate setting.

Such settings arise in attempts of understanding the "information paradox" of black holes, and in attempts of making sense of the idea of AdS/CFT correspondence (cf.\ e.g.\ \cite{EW} and references therein). In the latter context, MTTs run under the name "holographic screens". In fact an area theorem has been obtained in \cite{BE} (cf.\  Thm.\  IV.3) and \cite{ABHR} for so-called "regular holographic screens" whose definition consists of four somewhat subtle conditions (cf.\ Def.\   II.8 of \cite{BE}).
We have not been able to extract the required information from our Thm.\  \ref{family} to check these
requirements; in particular the causal character of the constructed MTT is unclear except that they  must  be spacelike near the turning points $\pm \psi_m$ where the mean curvatures achieve their extrema $\pm {\cal H}^{\mathfrak g}_{\psi_m}$ (cf.\ item 4 of Theorem \ref{family}).
In order to obtain monotonicity of area of the MTS (item 5 of Thm.\  \ref{family}) we are thus bound to the argument involving their rescaled Willmore energy (item 5 of Thm.\  \ref{thmunit}).

~\\


\noindent  {\large\bf Acknowledgement.}
M.M. acknowledges financial support under Projects PID2021-122938NB-I00 (Spanish Ministerio de Ciencia e Innovación and FEDER “A way of making Europe”),  SA097P24 (JCyL) and RED2022-134301-T (MCIN).

This research was also funded in part by the Austrian Science Fund
(FWF) [Grant  DOI 10.55776/P35078] (C.R. and W.S.) as well as [Grant DOI
10.55776/P33594] and [Grant DOI 10.55776/EFP6] (R.S.).

For open access purposes, the authors have applied a CC BY public copyright license to any author accepted manuscript version arising from this submission.


\end{document}